\documentclass[11pt]{article}
\usepackage[margin=1.05in]{geometry}
\usepackage{amsmath,amssymb,amsthm}
\usepackage{graphicx}
\usepackage{booktabs}
\usepackage{algorithm}
\usepackage{algpseudocode}
\usepackage[colorlinks=true,linkcolor=blue!60!black,citecolor=blue!60!black,urlcolor=blue!60!black]{hyperref}
\usepackage{xcolor}
\usepackage{graphicx}
\usepackage{pdflscape}
\newcommand{\figslot}[2]{\IfFileExists{#1}{\includegraphics[width=#2]{#1}}{\fbox{\parbox[c][0.55#2][c]{0.96#2}{\centering\small\texttt{\detokenize{#1}}}}}}

\newtheorem{theorem}{Theorem}
\newtheorem{lemma}{Lemma}
\newtheorem{proposition}{Proposition}
\newtheorem{corollary}{Corollary}
\theoremstyle{definition}
\newtheorem{definition}{Definition}

\newcommand{\supp}{\operatorname{supp}}
\newcommand{\card}{\operatorname{card}}
\newcommand{\Ipos}{I^{+}}
\newcommand{\Ineg}{I^{-}}
\newcommand{\Izero}{I^{0}}
\newcommand{\LPprune}{\mathrm{prune}}
\newcommand{\FMstep}{\mathrm{elim}}

\title{Accelerating Fourier--Motzkin elimination:\\ redundancy removal and the choice of variable elimination order}
\author{Shashaank Khanna\\[3pt]
\small Department of Mathematics, University of York, Heslington, York, YO10 5DD, United Kingdom\\
\small Aix-Marseille University, CNRS, LIS, Marseille, France\\
\small \texttt{shashaank.khanna@lis-lab.fr}}
\date{\today}

\begin{document}
\maketitle

\begin{abstract}
Fourier--Motzkin elimination computes an inequality description of the projection of a polyhedron onto a subset of its coordinates by eliminating one variable at a time. It is used in several areas of optimisation and computer science, and it is a standard way of obtaining the entropic constraints of a causal structure, where the marginalisation over the latent variables produces such a projection. Its limitation is the growth of the intermediate systems of inequalities, which can be doubly exponential in the number of eliminated variables even though the projection itself grows only as a single exponential. In practice the computational overload of the method therefore depends on two choices: how the redundant inequalities are removed after each step, and the order in which the variables are eliminated. We consider both. We first show, by an explicit example, that Imbert's redundancy test cannot be interleaved with redundancy removal by linear programming. We show that the two methods, however, can be combined soundly if the derivation records used by Imbert's test are re-initialised after every step at which linear programming is used. We then propose a rule for choosing the elimination order of the variables that gives a significant computational advantage, however, at the cost of increased resource usage. We demonstrate this advantage on some random polytopes, where the rule reduces the running time by factors of between $6$ and $25$ compared with the same elimination under a fixed order. For entropic descriptions of causal structures, with more than $250$ inequalities and more than $100$ variables to eliminate, our rule keeps the number of inequalities handled at each step one to two orders of magnitude lower than a fixed order. 
\end{abstract}

\section{Introduction}
\label{sec:intro}

Fourier--Motzkin \cite{fourier1824,motzkin1936} elimination computes the projection of a polyhedron by one variable at a time. Each step replaces the current system of inequalities by one in which the chosen variable no longer appears and whose solution set is the projection of the previous solution set along that variable. The method is simple to implement and it is used in integer programming~\cite{dantzig1972,williams1976}, in robust optimisation~\cite{zhen2018,yanikoglu2019}, in the integration of polynomials over polyhedra~\cite{baldoni2011,schechter1998}, in compiler optimisation~\cite{pouchet2008}, in circuit design~\cite{stehr2007}, and in the conversion between the vertex and the half-space descriptions of a polytope~\cite{ziegler1995}. 

Beyond the above applications, Fourier--Motzkin elimination is also used in the projection of an entropic polyhedral cone, corresponding to a causal structure, onto a subset of its variables. A causal structure is a directed acyclic graph over observed and latent nodes, and one of the basic questions about it is which correlations of the observed variables it allows~\cite{pearl2009,hlp2014,khanna2024}. One approach to this question works with entropies rather than probabilities~\cite{braunstein1988,fritzchaves2013,chaves2014,chaves2015,weilenmann2017}. The entropies of all the subsets of the variables satisfy Shannon's inequalities together with the conditional-independence equalities implied by the graph, and the entropic constraints on the observed variables alone are obtained by eliminating every entropy coordinate that involves a latent variable. This elimination is the projection of a polyhedral cone onto a subset of its coordinates. 

The limitation of Fourier--Motzkin (FM) elimination is well known. Eliminating a variable from a system of $m$ inequalities can produce up to $m^2/4$ inequalities, so $d$ successive eliminations can produce a number of inequalities that is doubly exponential in $d$~\cite{duffin1974,huynh1992}, even though the projection itself grows only as a single exponential (Lemma~\ref{lem:support} below; see also~\cite{monniaux2010}). Almost all of the inequalities produced are therefore redundant, i.e., implied by the others, and any implementation has to remove the redundant inequalities as it goes. Its running time then depends strongly on the order in which the variables are eliminated~\cite{huynh1992,simonking2005}.

The systems we need to project are often large. For e.g., a causal structure with four observed and three latent nodes already leads to a system of several hundred inequalities, and if from that system, say, $112$ of the $127$ entropy coordinates have to be eliminated (like in Section~\ref{sec:causal}), standard Fourier--Motzkin elimination becomes computationally very expensive, often intractable. To our knowledge, the entropic characterisations that have been computed so far concern structures with fewer nodes than this, and the cost of the elimination is what stands in the way of larger ones. In this work we consider the two choices on which the method depends in practice: how the redundant inequalities are removed, and in which order the variables are eliminated.

Two redundancy tests are in common use. The first solves one linear programme (LP) per inequality and decides its redundancy exactly, but each LP involves the whole of the current system. The second, due to Imbert~\cite{imbert1990,imbert1993}, inspects only the way each inequality was derived; it costs almost nothing, but detects only some of the all the redundancies. Since the two tests have complementary strengths, it is natural to apply Imbert's test after each elimination step and then the LP test to whatever inequalities survive. Our first result is that such a combination is not sound (Proposition~\ref{prop:incompat}): we give a simple system of four inequalities in four variables for which the above combination deletes every inequality of the projection. The reason is that Imbert's test certifies redundancy with respect to the complete set of inequalities generated by the elimination, whose record it uses in the next step, while the LP test removes members of this set itself, possibly including the one that makes the certificate valid. We then show that the two tests can be combined in any pattern, provided that the derivation records used by Imbert's test are re-initialised after every step at which the LP test is applied (Theorem~\ref{prop:sound}); the block alternation used in our code is a special case.

Our second result concerns the elimination order. The usual greedy rule eliminates the variable that produces the fewest inequalities. This count though includes the redundant inequalities, and on our instances the rule was much slower than a random order. We propose instead to eliminate every remaining variable tentatively, to prune each of the resulting systems by deleting any redundancies, and to keep the variable that leaves the fewest inequalities (Algorithm~\ref{alg:look-ahead}). The tentative eliminations are independent of one another, so they run in parallel, and the order found is itself worth keeping, since replaying it later costs a single ordinary run. On some random polytopes the rule reduced the wall-clock running time by factors of between $6$ and $25$ relative to the same elimination under a fixed order, and on some causal-structure instances, each with more than $250$ inequalities and more than $100$ variables to eliminate, it kept the number of inequalities handled at each step one to two orders of magnitude lower (Section~\ref{sec:results}) than compared to a fixed order of elimination of the variables.

The paper is organised as follows. Section~\ref{sec:fm} recalls the FM elimination step, illustrates it on some small examples, and states its correctness and the two growth bounds; the longer proofs are stated in Appendix~\ref{app:proofs}. Section~\ref{sec:redundancy} describes the two redundancy tests, gives the example showing that they cannot be interleaved, and gives the condition under which they can be combined. Section~\ref{sec:order} presents the rule for choosing the elimination order. Section~\ref{sec:causal} describes the systems arising from causal structures, Section~\ref{sec:results} reports the computational results and Section~\ref{sec:conclusion} concludes the paper.

\section{Fourier--Motzkin elimination}
\label{sec:fm}

\subsection{The elimination step}
\label{sec:fmstep}

Let $A\in\mathbb{R}^{m\times n}$ and $b\in\mathbb{R}^{m}$, and consider the system
\begin{equation}
Ay \le b, \qquad y\in\mathbb{R}^{n},
\label{eq:system}
\end{equation}
with the solution set $P=\{y\in\mathbb{R}^{n} : Ay\le b\}$. We write $a_{ik}$ for the entries of $A$ and $b_i$ for those of $b$. To eliminate the variable $y_s$, partition the row indices according to the sign of the coefficient of $y_s$ in the inequalities,
\begin{equation}
\Ipos=\{i: a_{is}>0\},\qquad \Ineg=\{j: a_{js}<0\},\qquad \Izero=\{l: a_{ls}=0\}.
\end{equation}
Dividing the rows in $\Ipos$ and $\Ineg$ by $|a_{is}|$, each inequality in $\Ipos$ becomes an upper bound on $y_s$ and each inequality in $\Ineg$ a lower bound,
\begin{align}
y_s &\le U_i(\hat y):=\frac{1}{|a_{is}|}\Big(b_i-\sum_{k\neq s}a_{ik}y_k\Big), && i\in\Ipos, \label{eq:upper}\\
y_s &\ge L_j(\hat y):=\frac{1}{|a_{js}|}\Big(\sum_{k\neq s}a_{jk}y_k-b_j\Big), && j\in\Ineg, \label{eq:lower}
\end{align}
where $\hat y=(y_1,\dots,y_{s-1},y_{s+1},\dots,y_n)$. The elimination step outputs the inequalities of $\Izero$, unchanged, together with the inequalities
\begin{equation}
L_j(\hat y)\ \le\ U_i(\hat y), \qquad i\in\Ipos,\ j\in\Ineg,
\label{eq:pairs}
\end{equation}
one for each pair. Each of these is free of $y_s$, so the output is a system in the $n-1$ variables $\hat y$ with $|\Ipos|\,|\Ineg|+|\Izero|$ inequalities. If $\Ipos$ is empty there are no pairs and the inequalities in $\Ineg$ are simply dropped, and likewise if $\Ineg$ is empty. If all three sets are empty the output is the empty system, whose solution set is $\mathbb{R}^{n-1}$.

\subsection{Correctness}
\label{sec:correctness}

Lemma~\ref{lem:step} and Corollary~\ref{cor:global} below are the standard correctness statements for the method, see, e.g., \cite[Chapter~12]{schrijver1986} and \cite[Lecture~1]{ziegler1995}; we include the short arguments because we use them repeatedly. Let $\pi_s:\mathbb{R}^{n}\to\mathbb{R}^{n-1}$ denote the map that deletes the $s$-th coordinate. The content of the elimination step is easy to state in words: the inequalities \eqref{eq:pairs} say that every lower bound on $y_s$ lies below every upper bound, which is the condition for a value of $y_s$ to exist. The following lemma makes this precise.

\begin{lemma}
\label{lem:step}
Let $P'\subseteq\mathbb{R}^{n-1}$ be the solution set of the system output by the elimination of $y_s$. Then $P'=\pi_s(P)$. In other words, $\hat y$ satisfies the new system if and only if there is a value of $y_s$ for which $(y_1,\dots,y_s, \dots,y_n)$ satisfies the original one.
\end{lemma}

\begin{proof}
Suppose first that $y\in P$. Each inequality in $\Izero$ holds at $y$ and does not involve $y_s$, so it holds at $\hat y$. Each inequality \eqref{eq:pairs} is obtained by dividing an inequality in $\Ipos$ and one in $\Ineg$ by positive numbers and adding the results, so it also holds at $y$, and since $y_s$ cancels in the sum it holds at $\hat y$. Hence $\pi_s(P)\subseteq P'$.

Conversely, suppose that $\hat y\in P'$. If $\Ipos$ and $\Ineg$ are both non-empty, let $L(\hat y)=\max_{j\in\Ineg}L_j(\hat y)$ and $U(\hat y)=\min_{i\in\Ipos}U_i(\hat y)$. The inequalities \eqref{eq:pairs} say precisely that $L_j(\hat y)\le U_i(\hat y)$ for all $i\in\Ipos$ and $j\in\Ineg$, i.e., that $L(\hat y)\le U(\hat y)$. Take any $y_s$ in the interval $[L(\hat y),U(\hat y)]$. Then \eqref{eq:upper} and \eqref{eq:lower} hold, and the inequalities in $\Izero$ hold because $\hat y\in P'$. If $\Ipos$ is empty, the original system constrains $y_s$ only from below, through \eqref{eq:lower}, and any $y_s\ge L(\hat y)$ will do; if $\Ineg$ is empty, any $y_s\le U(\hat y)$ will do; if both are empty, $y_s$ is arbitrary. In each case we have found $y_s$ with $(y_1,\dots,y_s, \dots,y_n)\in P$, so $P'\subseteq\pi_s(P)$.
\end{proof}

\begin{corollary}
\label{cor:global}
Eliminate the variables one at a time, in any order. After each step the current system describes the projection of $P$ onto the remaining variables. In particular:
\begin{enumerate}
\item[(i)] After all $n$ variables have been eliminated the system consists of inequalities between numbers, $0\le c_r$, and $P$ is non-empty if and only if $c_r\ge 0$ for all $r$.
\item[(ii)] The solutions of \eqref{eq:system} can be enumerated by back-substitution: choosing values for the variables one at a time, in the reverse of the elimination order, each within the interval between the largest lower bound and the smallest upper bound placed on it by the system at the corresponding stage (an infinite endpoint when there is no bound of that kind), produces exactly the points of $P$.
\end{enumerate}
\end{corollary}

The proof is short and is given in Appendix~\ref{app:proofs}.

As an illustration, consider the system
\begin{equation}
2x-3y\ge 1,\qquad 5x+3y\ge 1,\qquad 8x-3y\ge 1,\qquad x\ge 0,\qquad 3y\ge 0.
\label{eq:example}
\end{equation}
Eliminating $y$ combines each of the two inequalities in which $y$ has a positive coefficient with each of the two in which it has a negative coefficient, and copies $x\ge 0$. This gives $7x\ge 2$, $13x\ge 2$, $2x\ge 1$, $8x\ge 1$ and $x\ge 0$, of which only $2x\ge 1$ is non-redundant, so the projection of the solution set onto the $x$-axis is the ray $x\ge\frac12$. Taking $x=1$ and substituting back into \eqref{eq:example} gives $0\le y\le\frac13$, so, for example, $(x,y)=(1,0)$ is a solution, in accordance with Corollary~\ref{cor:global}(ii). Note also that the coefficient of $x$ is positive in every inequality of \eqref{eq:example} in which $x$ appears. Eliminating $x$ therefore forms no pairs, all inequalities involving $x$ are dropped, and the output is the single inequality $3y\ge 0$. This is correct: the solution set is unbounded in the direction of increasing $x$, and its projection onto the $y$-axis is $\{y\ge 0\}$. Two further examples, one in three dimensions and one in which infeasibility is detected through a violated numerical inequality, follow.

\subsection{Two further examples}
\label{sec:examples}

The projection computed by the elimination is easiest to see in three dimensions. Consider the system
\begin{equation}
-2\le x\le 2,\qquad -2\le y\le 2,\qquad -2\le z\le 2,\qquad 6x-y+2z\le 4,\qquad x+y-z\le -1,
\label{eq:box}
\end{equation}
whose solution set is the polyhedron shown in Fig.~\ref{fig:poly3d}. To eliminate $z$ we have $\Ipos=\{z\le 2,\ 6x-y+2z\le 4\}$, $\Ineg=\{-z\le 2,\ x+y-z\le -1\}$, and the four bounds on $x$ and $y$ in $\Izero$. The four pairs give $0\le 4$, $x+y\le 1$, $6x-y\le 8$ and $8x+y\le 2$, so the output is
\begin{equation}
-2\le x\le 2,\qquad -2\le y\le 2,\qquad 0\le 4,\qquad 8x+y\le 2,\qquad x+y\le 1,\qquad 6x-y\le 8,
\end{equation}
a system in $x$ and $y$ only. Its solution set, shown in Fig.~\ref{fig:poly2d}, is the shadow of the polyhedron of Fig.~\ref{fig:poly3d} on the $x$--$y$ plane, as Lemma~\ref{lem:step} requires. The inequality $0\le 4$ comes from the pair $z\le 2$, $-z\le 2$; it is redundant and is removed by either of the tests of Section~\ref{sec:redundancy}.

Infeasibility appears as a violated numerical inequality, as in Corollary~\ref{cor:global}(i). The system
\begin{equation}
x\ge 1,\qquad x\le -1,\qquad -1\le y\le 1
\end{equation}
has no solution. Eliminating $x$ combines the first two inequalities into $0\le -2$, which shows that the system, and hence each of its projections, is empty.

\begin{figure}[t]\centering
\begin{minipage}[t]{0.48\linewidth}\centering
\figslot{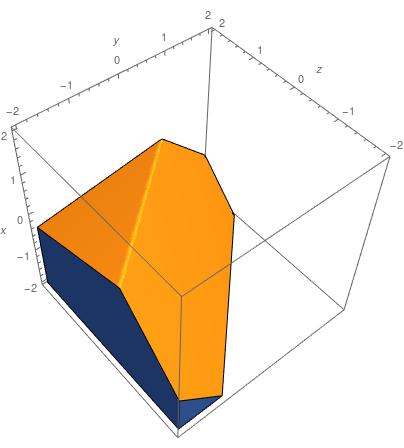}{\linewidth}
\caption{The polyhedron defined by \eqref{eq:box}.}
\label{fig:poly3d}
\end{minipage}\hfill
\begin{minipage}[t]{0.48\linewidth}\centering
\figslot{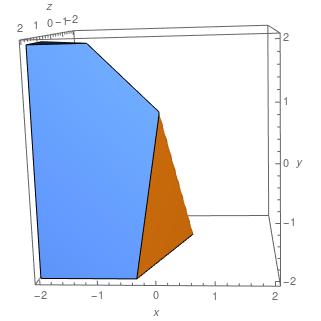}{\linewidth}
\caption{Its projection onto the $x$--$y$ plane, obtained by eliminating $z$.}
\label{fig:poly2d}
\end{minipage}
\end{figure}

\subsection{Growth of the intermediate systems}
\label{sec:growth}

The example above already contains redundant inequalities, and at scale these are the main obstacle to the computation.

\begin{proposition}
\label{prop:growth}
An elimination step applied to a system of $m$ inequalities produces at most $\max(m,\lfloor m^{2}/4\rfloor)$ inequalities. Consequently, if $m\ge 4$, after $d$ elimination steps the system has at most $4\,(m/4)^{2^{d}}$ inequalities.
\end{proposition}

The proof is elementary and is given in Appendix~\ref{app:proofs}. The doubly exponential bound concerns the number of inequalities generated, not the complexity of the projection. The next lemma shows that the projection always has a description of singly exponential size. The description of the projection by the extreme rays of the cone that appears in its proof is due to Balas~\cite{balas1998} (see also~\cite{monniaux2010}), and the bound on the number of non-zero entries of an extreme ray is the usual characterisation of extreme rays, see, e.g., \cite[Chapter~8]{schrijver1986}. It is also the fact behind Imbert's test in Section~\ref{sec:imbert}. For a set $D$ of variables we write $\pi_D$ for the map deleting the coordinates in $D$.

\begin{lemma}
\label{lem:support}
Let $P\neq\emptyset$ and let $D$ be a set of $d$ of the variables. Every inequality in the variables outside $D$ that is valid on $\pi_D(P)$ is implied by valid inequalities each of which is a non-negative combination of at most $d+1$ of the inequalities of \eqref{eq:system}. In particular, $\pi_D(P)$ has a description consisting of at most $\sum_{k=1}^{d+1}\binom{m}{k}$ inequalities, and if $\pi_D(P)$ is full-dimensional then every non-redundant description of it has at most this many inequalities.
\end{lemma}

The idea of the proof is simple. By Farkas' lemma, every inequality valid on the projection is a non-negative combination of the original inequalities in which the coefficients of the eliminated variables cancel. The multiplier vectors of such combinations form a polyhedral cone, every element of which is a non-negative combination of the extreme rays of the cone, and a dimension count shows that an extreme ray has at most $d+1$ non-zero entries. The inequalities given by the extreme rays therefore describe the projection, and there is at most one extreme ray for each choice of the set of non-zero entries, which gives the count. The details are in Appendix~\ref{app:proofs}.

Applied to the intermediate projections, the lemma shows that if redundancies are removed completely after each step, then after $t$ steps the system has at most $\sum_{k=1}^{t+1}\binom{m}{k}=O(m^{t+1})$ inequalities whenever the projection at that stage is full-dimensional.\footnote{If a projection is not full-dimensional, a non-redundant description of it may contain additional inequalities which together encode its implicit equalities. We do not need a bound in this case.} The contrast between Proposition~\ref{prop:growth} and Lemma~\ref{lem:support} is the reason that redundancy removal is worthwhile: without it the method handles doubly exponentially many inequalities in order to describe an object whose description is only singly exponential. It remains to decide how the redundancies should be removed and in which order the variables should be eliminated, and we take these in turn now.

\section{Redundancy removal}
\label{sec:redundancy}

\subsection{The linear-programming test}
\label{sec:lptest}

An inequality $a^{\top}y\le\beta$ of a system $S$ is redundant in $S$ if it is implied by the other inequalities of $S$, i.e., if $\max\{a^{\top}y : y \text{ satisfies } S\setminus\{a^{\top}y\le\beta\}\}\le\beta$. This maximum is the optimal value of a linear programme, so redundancy can be decided exactly by solving one LP per inequality.\footnote{If the maximum is unbounded the inequality is certainly not redundant. Some LP solvers report an unbounded problem in free variables as infeasible, so an implementation has to check the status returned by the solver and not only the value.} Two identical inequalities are each redundant given the other, so the inequalities have to be tested one at a time against the system of those retained so far, each redundant inequality being deleted before the next is tested. This yields a subsystem with the same solution set in which no inequality is redundant. Each LP involves the whole of the current system, and after an elimination step the number of inequalities to be tested can be of order $m^2/4$. The LP test is therefore exact but expensive.

\subsection{Imbert's test}
\label{sec:imbert}

Imbert~\cite{imbert1990,imbert1993} introduced a test that inspects only the way in which each inequality was derived. Fix an initial system, which we call the root, and run the elimination while maintaining the following records for every inequality $i$ of the current system.

\begin{definition}
\label{def:imbert}
The \emph{history} $H_i$ of an inequality $i$ is the set of root inequalities used in its derivation: the history of a root inequality is the inequality itself, a copied inequality inherits its history, and the history of an inequality formed from a pair is the union of the histories of the two members of the pair. The \emph{explicit variables} $E_i$ are the variables that were eliminated by forming pairs somewhere in the derivation of $i$. The \emph{implicit variables} $I_i$ are the variables that occur in at least one member of $H_i$, but do not occur in $i$, and are not in $E_i$. Finally $O_k$ denotes the set of variables eliminated in the first $k$ steps.
\end{definition}

In words, the history records which of the root inequalities an inequality descends from, and the other records keep track of which variables have disappeared from it along the way.

\begin{theorem}[Imbert~\cite{imbert1990,imbert1993}]
\label{thm:imbert}
If, after $k$ elimination steps from the root system, an inequality $i$ of the current system violates
\begin{equation}
\card(H_i)\ \le\ 1+\card\!\big(E_i\cup(I_i\cap O_k)\big),
\label{eq:imbertcond}
\end{equation}
then $i$ is redundant in the system of inequalities generated by the elimination at that stage.
\end{theorem}

Condition \eqref{eq:imbertcond} is the counterpart of Lemma~\ref{lem:support} at the level of individual derivations. Since $E_i$ and $I_i\cap O_k$ are subsets of $O_k$, the right-hand side is at most $k+1$, so an inequality whose history contains more than $k+1$ root inequalities is always flagged; by Lemma~\ref{lem:support} no such inequality is ever needed to describe the projection, and the theorem says that it is implied by the other inequalities produced by the elimination. The sets $E_i$ and $I_i$ refine the count $k$ by taking into account variables that disappear from an inequality without having been eliminated explicitly; we refer to~\cite{imbert1990} for the details. The test is sufficient for redundancy but not necessary, and it costs almost nothing to apply. There is a second theorem due to Imbert, but we do not consider it here since it guarantees neither redundancy nor non-redundancy of an inequality.

\subsection{The two tests cannot be interleaved}
\label{sec:incompat}

Since Imbert's test is cheap and the LP test is expensive, it is natural to apply Imbert's test after each elimination step and then to apply the LP test to the inequalities that survive, in the hope of solving fewer and smaller LPs. The following proposition shows that this is not sound.

\begin{proposition}
\label{prop:incompat}
There is a system of linear inequalities and an elimination order for which the following procedure outputs a system whose solution set strictly contains the projection: after each elimination step, delete the inequalities that violate \eqref{eq:imbertcond}, then delete the inequalities that the LP test finds redundant, and carry the records of Definition~\ref{def:imbert} over to the next step. In the example below the procedure outputs the empty system.
\end{proposition}

\begin{proof}
Consider the following system in the variables $(w,x,y,z)$, which are otherwise unconstrained:
\begin{equation}
\begin{aligned}
(\mathrm{E}1)\quad -x+w &\le 0, &\qquad (\mathrm{E}2)\quad x+y &\le 0,\\
(\mathrm{E}3)\quad -x+y+z-w &\le 0, &\qquad (\mathrm{E}4)\quad x-y-z+w &\le 0.
\end{aligned}
\label{eq:counterexample}
\end{equation}
We eliminate $x$ and then $z$. The projection onto $(w,y)$ is the half-plane $\{y+w\le 0\}$. Indeed, adding (E1) and (E2) shows that $y+w\le 0$ holds on the solution set, and conversely, if $y+w\le 0$ then $x=w$ and $z=2w-y$ satisfy all four inequalities.

\emph{Elimination of $x$.} Here $\Ipos=\{\mathrm{E}2,\mathrm{E}4\}$, $\Ineg=\{\mathrm{E}1,\mathrm{E}3\}$ and $\Izero=\emptyset$, and the four pairs give
\begin{equation}
\begin{aligned}
(\mathrm{F}1)\quad y+w&\le 0, & H&=\{\mathrm{E}1,\mathrm{E}2\},\\
(\mathrm{F}2)\quad 2y+z-w&\le 0, & H&=\{\mathrm{E}2,\mathrm{E}3\},\\
(\mathrm{F}3)\quad -y-z+2w&\le 0, & H&=\{\mathrm{E}1,\mathrm{E}4\},\\
(\mathrm{F}4)\quad 0&\le 0, & H&=\{\mathrm{E}3,\mathrm{E}4\},
\end{aligned}
\label{eq:afterx}
\end{equation}
each with $E=\{x\}$. Imbert's test flags nothing, since $\card(H)=2\le 1+\card(\{x\})$ for each of the four. The LP test, on the other hand, finds (F1) redundant, because the sum of (F2) and (F3) is $y+w\le 0$, and it finds (F4) redundant trivially.

\emph{Imbert's test alone, LP test at the end.} Suppose we keep all four inequalities of \eqref{eq:afterx} and eliminate $z$. Then (F2) and (F3) form a pair, giving a second copy of $y+w\le 0$ with $H=\{\mathrm{E}1,\mathrm{E}2,\mathrm{E}3,\mathrm{E}4\}$ and $E=\{x,z\}$, while (F1) and (F4) are copied. Condition \eqref{eq:imbertcond} fails for the new copy, since $4\not\le 1+2$, and Imbert's test deletes it. This is correct: the new copy is redundant given the copy (F1). A final LP pass then leaves $\{y+w\le 0\}$, which is the projection.

\emph{The interleaved procedure.} Now suppose that after the elimination of $x$ we delete (F1) and (F4), as the LP test instructs, leaving (F2) and (F3). Eliminating $z$ produces the single inequality $y+w\le 0$ with $H=\{\mathrm{E}1,\mathrm{E}2,\mathrm{E}3,\mathrm{E}4\}$ and $E=\{x,z\}$. Condition \eqref{eq:imbertcond} fails, the inequality is deleted, and the output is the empty system, whose solution set is $\mathbb{R}^2$.
\end{proof}

This failure is not particular to the example. Theorem~\ref{thm:imbert} asserts that a flagged inequality is redundant in the system of inequalities generated by the elimination, and the proof of the theorem relies on the presence of the other inequalities generated from the same root. The LP test removes inequalities from this reference system. In the example, the inequality that is deleted at the second step is redundant given (F1), but (F1) has already been removed, and the certificate provided by \eqref{eq:imbertcond} no longer refers to anything that is present. Any schedule in which the records of Definition~\ref{def:imbert} are carried across an LP deletion is exposed to the same failure. The code accompanying this paper contains further examples.

\subsection{A sound combination}
\label{sec:sound}

The two tests can be combined provided they never share records. The precise condition is the following.

\begin{algorithm}[t]
\caption{Alternation of Imbert's test and the LP test with epoch length $T$.}
\label{alg:alternation}
\begin{algorithmic}[1]
\While{variables remain to be eliminated}
  \State declare the current system to be the root and initialise the records of Definition~\ref{def:imbert}
  \For{$t=1,\dots,T$, while variables remain}
     \State eliminate the next variable, updating the records
     \State delete every inequality that violates \eqref{eq:imbertcond}
  \EndFor
  \If{variables remain}
     \State eliminate the next variable
     \State delete redundant inequalities with the LP test
  \EndIf
\EndWhile
\State delete redundant inequalities with the LP test \Comment{Imbert's test alone is not complete}
\end{algorithmic}
\end{algorithm}

\begin{theorem}
\label{prop:sound}
Run the elimination in any order and, after each step, apply Imbert's test, or the LP test, or the first followed by the second, in any pattern, subject to one rule: after every step at which the LP test has been applied, discard the records of Definition~\ref{def:imbert} and declare the current system to be the root for the steps that follow. Then after every step the current system describes the projection of $P$ onto the remaining variables, and the final system, after a last application of the LP test, is a non-redundant description of the projection.
\end{theorem}

\begin{proof}
Call an epoch a maximal sequence of consecutive steps ending with a step at which the LP test is applied, or with the last step. We show by induction on the epochs that at the beginning of each epoch the current system describes the projection of $P$ onto the variables remaining at that point. This holds for the first epoch. Within an epoch, each elimination step preserves the projection by Lemma~\ref{lem:step}, and each deletion by Imbert's test removes an inequality that, by Theorem~\ref{thm:imbert} applied to the root of the epoch, is redundant in the system generated from that root, exactly as in Imbert's algorithm~\cite{imbert1990}, in which flagged inequalities are deleted as they arise. Hence throughout the epoch the system describes the projection of the solution set of the root, which by the induction hypothesis is the projection of $P$. The deletions by the LP test at the end of the epoch preserve the solution set, so the system at the beginning of the next epoch, which is the new root, describes the correct projection. The final application of the LP test preserves the solution set and, being sequential, leaves no redundant inequality (Section~\ref{sec:lptest}).
\end{proof}

Proposition~\ref{prop:incompat} shows that the rule in Theorem~\ref{prop:sound} cannot simply be dropped. Note also that Imbert's test can flag nothing at the first step after a re-initialisation, because every history then has at most two elements, so there is no point in applying the LP test at every step. Algorithm~\ref{alg:alternation} is the pattern we use: Imbert's test alone for $T$ steps, then one step with the LP test, then re-initialisation. Running Imbert's test alone with the LP test only at the end, as in the first run in the proof of Proposition~\ref{prop:incompat}, is the case of a single epoch. In our experiments we used $T=3$. On some instances Algorithm~\ref{alg:alternation} is faster than the LP test alone, but on the instances of Section~\ref{sec:causal} it is not: Imbert's test detects too few of the redundancies, the intermediate systems grow, and the LP steps become more expensive than they would have been had they been applied throughout. This is what led us to consider the elimination order of the variables.

\section{Choosing the elimination order}
\label{sec:order}

\subsection{The greedy rule}
\label{sec:greedy}

The number of inequalities produced by an elimination step depends on which variable is eliminated. In the notation of Section~\ref{sec:fmstep}, eliminating $y_s$ produces
\begin{equation}
\nu(s)=|\Ipos_s|\,|\Ineg_s|+|\Izero_s|
\label{eq:nu}
\end{equation}
inequalities, and this number varies considerably from variable to variable. In the system \eqref{eq:example}, for instance, eliminating $y$ produces $\nu=2\cdot 2+1=5$ inequalities whereas eliminating $x$ produces $\nu=4\cdot 0+1=1$. The quantity $\nu(s)$ costs $O(m)$ to evaluate, and the usual greedy rule eliminates at each step a variable minimising it~\cite{huynh1992}.

On large instances this rule was way slower for us than eliminating the variables in a random order. The reason is that $\nu(s)$ counts the inequalities produced by the next step, redundant ones included, and takes no account of the effect of the choice on the steps that follow. On our instances, choosing the variable that produces the fewest inequalities now routinely leads to systems with many more non-redundant inequalities later. An elimination order has to be judged by the total cost it induces, and the size of the next system is a poor proxy for that cost.

\subsection{Lookahead on the number of non-redundant inequalities}
\label{sec:look-ahead}

We therefore replace $\nu(s)$ by the quantity that determines the cost of the subsequent steps, namely the number of inequalities that remain after the redundancies have been removed. For the current system $S$ and a candidate variable $s$, let
\begin{equation}
\hat\nu(s)=\big|\LPprune\big(\FMstep(S,s)\big)\big|,
\label{eq:nuhat}
\end{equation}
where $\FMstep(S,s)$ is the system obtained by eliminating $s$ and $\LPprune$ removes redundant inequalities with the LP test. Evaluating $\hat\nu(s)$ requires carrying out the elimination and the pruning, so choosing the next variable by this rule costs $r$ tentative eliminations, where $r$ is the number of variables still to be eliminated. These tentative eliminations are independent of one another, and we run them in parallel, one per core, and keep the smallest resulting system and discard the others (Algorithm~\ref{alg:look-ahead}). 

\begin{algorithm}[t]
\caption{Choice of the elimination order by one-step look-ahead.}
\label{alg:look-ahead}
\begin{algorithmic}[1]
\State $S\gets$ the initial system with redundancies removed;\quad $D\gets$ the set of variables to eliminate;\quad $\sigma\gets()$
\While{$D\neq\emptyset$}
  \ForAll{$s\in D$, in parallel}
     \State $S_s\gets \LPprune(\FMstep(S,s))$
  \EndFor
  \State $s^{*}\gets\arg\min_{s\in D}|S_s|$;\quad $S\gets S_{s^{*}}$;\quad $D\gets D\setminus\{s^{*}\}$;\quad append $s^{*}$ to $\sigma$
\EndWhile
\State \Return the projection $S$ and the elimination order $\sigma$
\end{algorithmic}
\end{algorithm}

\subsection{Cost}
\label{sec:cost}

The rule multiplies the amount of computation per step by at most $r$. If $r$ cores are available the wall-clock time per step is unchanged; with fewer cores the tentative eliminations are distributed over the cores available. The rule is nevertheless worthwhile on our instances, and by a wide margin, because the cost of a step grows faster than linearly with the size of the current system: a system of $m$ inequalities can produce $m^2/4$ inequalities, each of which has to be tested by an LP over a system of comparable size. Keeping every intermediate system small therefore reduces the cost of every subsequent step, and this saving outweighs the cost of the look-ahead.

We make no claim of optimality. The rule looks only one step ahead, and there will be instances on which a deeper look-ahead, or a different order altogether, does better; nor can the worst-case behaviour of Proposition~\ref{prop:growth} be excluded.

\subsection{Reusing the order}
\label{sec:reuse}

The order $\sigma$ returned by Algorithm~\ref{alg:look-ahead} deserves separate mention. It is a list of the eliminated variables, and once it is known the projection can be recomputed by an ordinary sequential run: eliminate the variables in the order $\sigma$, removing the redundant inequalities after each step, on a single core and with no search. The tentative eliminations, which are what consume the parallel resources, are thus a cost of finding $\sigma$ and not a cost of using it. This matters for two reasons. In the application of Section~\ref{sec:causal} the same causal structure is often marginalised more than once, for instance with different sets of additional constraints, and the order found in one of these computations is the natural first thing to try in the others. And replaying $\sigma$ on one core is the control experiment for the method: its running time measures the quality of the order on its own, separately from the cost of finding it, and it is the number against which the fixed-order times of Section~\ref{sec:results} should ultimately be compared.

\section{Entropic constraints from causal structures}
\label{sec:causal}

We now describe the systems that led us to the method. They arise in the entropic approach to causal structures, which goes back to the information-theoretic Bell inequalities of Braunstein and Caves~\cite{braunstein1988} and was developed into a general method in~\cite{fritzchaves2013,chaves2014,chaves2015,weilenmann2017}; we summarise it only to the extent needed to define the computational problem. For the role of inequality constraints in the classification of causal structures see also~\cite{khanna2024, khanna2025closing, khanna2026spurious, hlp2014}.

Let $X_1,\dots,X_N$ be discrete random variables. The Shannon entropies $H(X_S)$ of the non-empty subsets $S\subseteq\{1,\dots,N\}$ form a vector in $\mathbb{R}^{2^N-1}$, the entropy vector of the distribution~\cite{shannon1948,yeung2002}. Every entropy vector satisfies the elemental inequalities~\cite{yeung1997}, namely the monotonicity relations $H(X_i\,|\,X_{\{1,\dots,N\}\setminus\{i\}})\ge 0$ and the submodularity relations $I(X_i\!:\!X_j\,|\,X_K)\ge 0$ for $i\neq j$ and $K\subseteq\{1,\dots,N\}\setminus\{i,j\}$, where the conditional entropy and the conditional mutual information are the linear combinations $H(Y|Z)=H(YZ)-H(Z)$ and $I(Y\!:\!Z|W)=H(YW)+H(ZW)-H(W)-H(YZW)$ of subset entropies. There are
\begin{equation}
N+\binom{N}{2}\,2^{\,N-2}
\label{eq:elementalcount}
\end{equation}
elemental inequalities, and the cone they define, the Shannon cone, contains the closure of the set of entropy vectors. For $N\ge 4$ the containment is strict~\cite{zhangyeung1997,zhangyeung1998}. This is a limitation of the entropic method that is inherited by everything computed from the Shannon cone, but it does not affect the polyhedral computation itself.

A causal structure is a directed acyclic graph whose nodes are either observed or latent~\cite{pearl2009}. The conditional independences implied by the graph, which can be read off it by $d$-separation, are linear equalities $I(X\!:\!Y|Z)=0$ on the entropy vector. Adding these equalities to the elemental inequalities and then eliminating every coordinate $H(X_S)$ for which $S$ contains a latent node gives a system of entropic inequalities that is satisfied by the observed variables of every classical model of the causal structure. This is the entropic analogue of marginalisation, and it is a Fourier--Motzkin problem. (It is best to use the equalities first, to substitute variables away at no cost, so that the remaining task is the elimination of variables from a system of inequalities.) With $o$ observed and $\ell$ latent nodes the number of coordinates to be eliminated is
\begin{equation}
2^{\,o+\ell}-2^{\,o},
\label{eq:elimcount}
\end{equation}
which grows exponentially. A structure with four observed and three latent nodes already requires the elimination of $2^7-2^4=112$ of the $127$ coordinates, starting from the $679$ elemental inequalities of \eqref{eq:elementalcount} together with the equalities of the structure. Systems of this size are beyond a straightforward implementation of the method, and it is for these that the elimination order matters most.

\section{Computational results}
\label{sec:results}

All experiments use our Python implementation, which contains the elimination step, the sequential LP test, Imbert's test with the schedule of Algorithm~\ref{alg:alternation}, and the look-ahead rule of Algorithm~\ref{alg:look-ahead} with one tentative elimination per worker. The code, together with the benchmark instances and the scripts that produce the figure and the table, is available from the author.%
\footnote{A public repository will soon be added to an updated version of this preprint.}
In each experiment the comparison is with the same elimination, with the same LP test, run under a fixed order, which is what a straightforward implementation of the method does.

\subsection{Causal structures}
\label{sec:resultscausal}

We consider four instances of the kind described in Section~\ref{sec:causal}, each with more than $250$ non-redundant inequalities initially and more than $100$ variables to eliminate; the causal structures and the initial systems are specified in the code. For these instances we record, at every elimination step, the number of redundant inequalities that were produced and the number of non-redundant inequalities that were retained, both for the order found by Algorithm~\ref{alg:look-ahead} and for the fixed order. Since the time taken by a step is dominated by the LPs solved to remove the redundant inequalities, and the number of LPs solved equals the number of inequalities produced, these counts determine the running time.

Figure~\ref{causal} shows the results. Under the fixed order the number of redundant inequalities rises repeatedly to one or two orders of magnitude above the number of non-redundant ones, reaching values of order $10^4$ in the first instance, and every one of these inequalities has to be processed by an LP before it is discarded. Under the look-ahead order the number of non-redundant inequalities decreases almost monotonically and the number of redundant inequalities stays within a small factor of it throughout.

\subsection{Random polytopes}
\label{sec:resultsrandom}
Table~\ref{tab:random} gives the total running times for six random polytopes in $15$ variables, from each of which $12$ variables were eliminated. The random instances were generated by the routine included with the code. The look-ahead rule is faster than the fixed order by factors of between $6$ and $25$.

\begin{table}[t]
\centering
\small
\setlength{\tabcolsep}{5pt}
\begin{tabular}{ccccccc}
\toprule
Polytope & \begin{tabular}{@{}c@{}}Initial\\irredundant\end{tabular} & \begin{tabular}{@{}c@{}}Final\\irredundant\end{tabular} & \begin{tabular}{@{}c@{}}Variables\\(total / eliminated)\end{tabular} & \begin{tabular}{@{}c@{}}Fixed order\\time (s)\end{tabular} & \begin{tabular}{@{}c@{}}Algorithm~\ref{alg:look-ahead}\\time (s)\end{tabular} & Ratio \\
\midrule
1 & 52 & 14 & 15 / 12 & 55496.4560 & 4126.1571 & 13.4 \\
2 & 44 & 13 & 15 / 12 & 18152.0970 & 2536.8207 & \phantom{0}7.2 \\
3 & 46 & 14 & 15 / 12 & \phantom{0}3620.4586 & \phantom{0}142.7669 & 25.4 \\
4 & 55 & \phantom{0}5 & 15 / 12 & \phantom{00}301.6454 & \phantom{00}49.0690 & \phantom{0}6.1 \\
5 & 56 & 11 & 15 / 12 & 19652.5220 & 1164.9615 & 16.9 \\
6 & 53 & 10 & 15 / 12 & 38530.8608 & 6513.9060 & \phantom{0}5.9 \\
\bottomrule
\end{tabular}
\caption{Wall-clock running times for six random polytopes: Fourier--Motzkin elimination with the sequential LP test under a fixed order, and Algorithm~\ref{alg:look-ahead}. The last column is the ratio of the two times. See Section~\ref{sec:scope} for the resources used by the two methods.}
\label{tab:random}
\end{table}

\begin{figure}[p]
    \centering
    \includegraphics[width=1\textwidth]{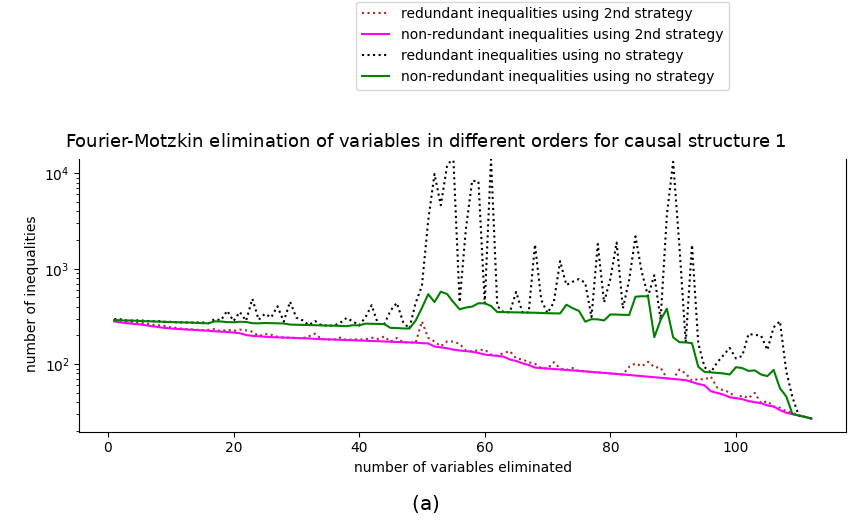}
    \vspace{1cm}

    \vfill

    \includegraphics[width=1\textwidth]{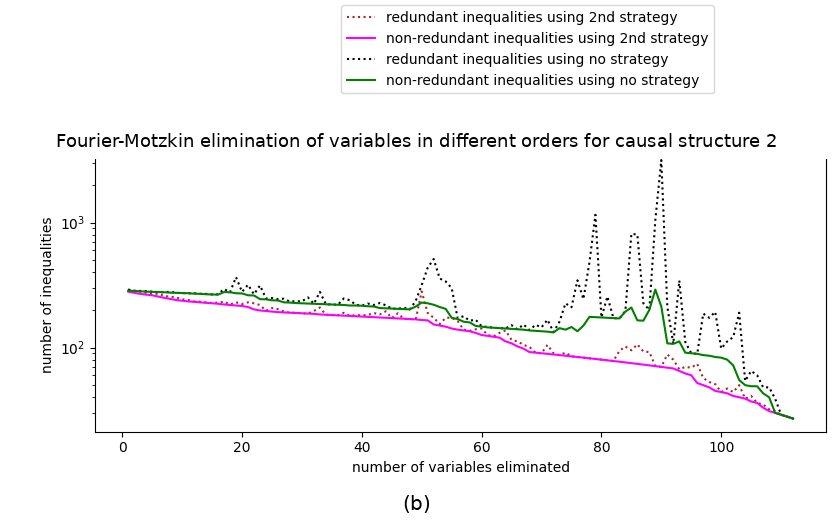} 
\end{figure}

\begin{figure}[p]
    \centering
    \includegraphics[width=1\textwidth]{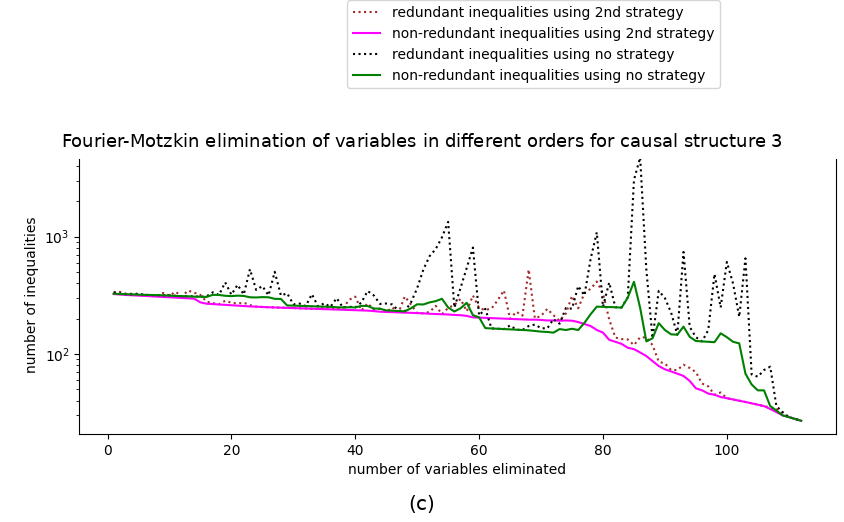} 
    \vspace{0.5cm}
    \vfill

    \includegraphics[width=1\textwidth]{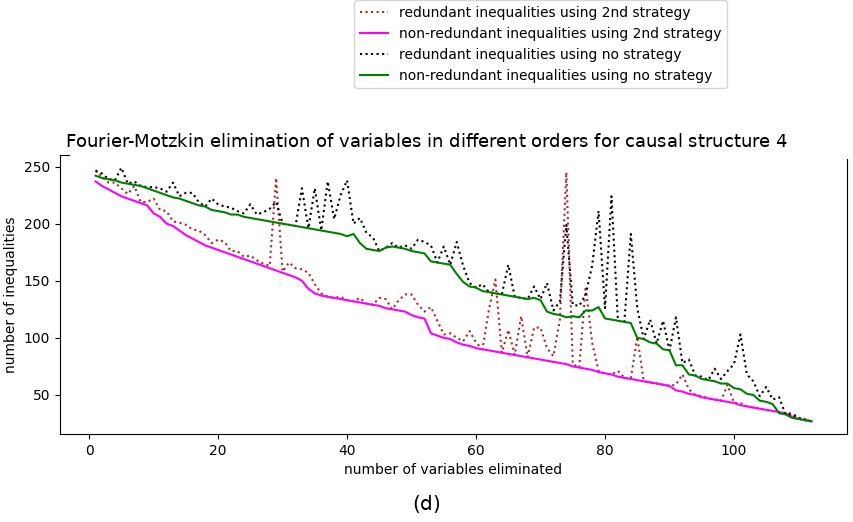} 

    \caption{Number of redundant inequalities produced and number of non-redundant inequalities retained at each elimination step, for four systems arising from causal structures, under the order found by Algorithm~\ref{alg:look-ahead} and under a fixed order. Panels (a)--(d) correspond to the four instances.}
    \label{causal}
\end{figure}

\subsection{Remarks on the comparison}
\label{sec:scope}

Three remarks are in order. First, the times are wall-clock times, and the two methods use different resources. The fixed-order run uses a single core, whereas Algorithm~\ref{alg:look-ahead} distributes up to $r$ tentative eliminations per step over the cores available. The total processor time of the look-ahead rule therefore exceeds its wall-clock time by up to a factor of $r$ per step, and on a machine with few cores the gain will be correspondingly smaller. The way to separate the quality of the order from the cost of finding it is to replay $\sigma$ on one core, which costs a single run (Section~\ref{sec:reuse}). Second, the rule is a heuristic, and Section~\ref{sec:cost} states what it does not promise. Third, for the causal-structure instances the system obtained is, as always in the entropic method, an outer approximation to the set of entropy vectors of the observed variables (Section~\ref{sec:causal}). This is a property of the entropic method and not of the elimination.

\section{Conclusion}
\label{sec:conclusion}

The usefulness of Fourier--Motzkin elimination depends on how the redundant inequalities it produces are managed. We have shown that the two standard tests for redundancy cannot be interleaved, because Imbert's test certifies redundancy with respect to the complete set of generated inequalities and the LP test removes elements of that set, and we have shown that they combine soundly provided the derivation records are re-initialised after every step at which the LP test is used. On our instances, however, the elimination order matters more than the choice of test. Choosing the next variable by the number of inequalities it produces is counterproductive, while choosing it by the number of non-redundant inequalities it leaves, at the cost of a one-step look-ahead that runs in parallel, reduced the running times by factors of up to $25$ and made the entropic marginalisation of causal structures with more than a hundred eliminated coordinates practicable. The order found can be replayed on a single core and reused in related computations.

Several questions remain. A look-ahead of depth $k$ costs $r^{k}$ tentative eliminations per step, and it would be interesting to know on which instances the extra cost pays. Since the orders found on structurally similar instances can be compared, one could also try to learn an ordering rule from them. The complexity of finding an optimal elimination order is, as far as we know, open; the analogous problem of minimising fill-in in sparse Gaussian elimination is NP-complete~\cite{yannakakis1981}, which suggests that heuristics such as ours are the appropriate tool. Finally, Lemma~\ref{lem:support} suggests that there may be tests for redundancy based on the derivation records that are stronger than Imbert's, and Proposition~\ref{prop:incompat} makes precise what such a test must respect if it is to be used alongside the LP test.

\section*{Additional Note}
Preliminary results appeared in the author's PhD thesis \cite{khanna2025exploring}. This a preliminary version of this article and some results will be added in the updated version, along with the entire code being made publicly available.

\section*{Acknowledgements}
The author thanks Matthew Pusey and Roger Colbeck for discussions. This work was supported by a studentship from the Department of Mathematics, University of York.

\appendix

\section{Proofs}
\label{app:proofs}

\begin{proof}[Proof of Corollary~\ref{cor:global}]
The first statement follows from Lemma~\ref{lem:step} by induction on the number of eliminated variables, since the projection onto the remaining coordinates is the composition of the projections along the eliminated ones. For (i), note that after all variables have been eliminated the projection is a subset of $\mathbb{R}^0$, which is either the single point of $\mathbb{R}^0$ or empty, and the point survives precisely when all the numerical inequalities hold.

For (ii), suppose that the variables are eliminated in the order $y_n,y_{n-1},\dots,y_1$, relabelling them if necessary, and let $P_k$ denote the projection of $P$ onto $(y_1,\dots,y_k)$, so that $P_n=P$. Fix $(y_1,\dots,y_k)\in P_k$. The system at the stage at which $y_{k+1}$ is about to be eliminated describes $P_{k+1}$, and the proof of Lemma~\ref{lem:step}, applied to that system, shows that the values of $y_{k+1}$ for which $(y_1,\dots,y_{k+1})\in P_{k+1}$ form the non-empty interval between the largest lower bound and the smallest upper bound that the system places on $y_{k+1}$ at the point $(y_1,\dots,y_k)$, with an infinite endpoint when there is no bound of the corresponding kind. A sequence of choices within these intervals therefore produces a point of $P$. Conversely, if $y^{*}\in P$ then $(y^{*}_1,\dots,y^{*}_k)\in P_k$ for every $k$, so $y^{*}_{k+1}$ lies in the interval determined by $(y^{*}_1,\dots,y^{*}_k)$ at every stage, and $y^{*}$ is produced.
\end{proof}

\begin{proof}[Proof of Proposition~\ref{prop:growth}]
Let $u=|\Ipos|$, $v=|\Ineg|$ and $w=|\Izero|$, so that $u+v+w=m$ and the output has $uv+w$ inequalities. If $w=m$ the output has $m$ inequalities. Otherwise $uv\le (u+v)^2/4=(m-w)^2/4$, so the output has at most $f(w):=(m-w)^2/4+w$ inequalities. The function $f$ is decreasing on $[0,m-2]$, where it is at most $f(0)=m^2/4$, and $f(m-1)=m-\tfrac34<m$. This proves the first claim. For the second, let $B_d=4(m/4)^{2^d}$, so that $B_0=m$, $B_{d+1}=B_d^2/4$ and $B_d\ge 4$ for all $d$ when $m\ge 4$. If the system after $d$ steps has $m_d\le B_d$ inequalities, then by the first claim the system after $d+1$ steps has at most $\max(m_d,m_d^2/4)\le\max(B_d,B_d^2/4)=B_d^2/4=B_{d+1}$ inequalities.
\end{proof}

\begin{proof}[Proof of Lemma~\ref{lem:support}]
Let $c^{\top}y\le\gamma$ be an inequality in the variables outside $D$ that is valid on $\pi_D(P)$. Regarded as an inequality on $\mathbb{R}^n$ with $c_D=0$, it is valid on $P$, and since $P\neq\emptyset$ the affine form of Farkas' lemma~\cite[Chapter~7]{schrijver1986} gives $\lambda\ge 0$ with $\lambda^{\!\top}\!A=c^{\top}$ and $\lambda^{\!\top}b\le\gamma$. Such a $\lambda$ lies in the cone $\Lambda_D=\{\lambda\in\mathbb{R}^m:\lambda\ge 0,\ (\lambda^{\!\top}\!A)_D=0\}$, the projection cone of~\cite{balas1998}. This is a polyhedral cone contained in the non-negative orthant, so it is pointed, and hence it is the set of non-negative combinations of its finitely many extreme rays. Let $\rho$ span an extreme ray of $\Lambda_D$, write $\supp\rho=\{i:\rho_i\neq 0\}$ and let $\sigma=|\supp\rho|$. The constraints of $\Lambda_D$ that hold with equality at $\rho$ are the $d$ equations $(\rho^{\top}\!A)_D=0$ and the $m-\sigma$ equations $\rho_i=0$ for $i\notin\supp\rho$. For $\rho$ to span an extreme ray these constraints must have rank $m-1$ (see, e.g., \cite[Chapter~8]{schrijver1986}), so $m-1\le d+(m-\sigma)$, i.e., $\sigma\le d+1$. Now write $\lambda=\sum_r\alpha_r\rho^{(r)}$ with $\alpha_r\ge 0$ and extreme rays $\rho^{(r)}$. The inequalities $(\rho^{(r)})^{\top}\!Ay\le(\rho^{(r)})^{\top}b$ are valid on $P$, involve only the variables outside $D$, are non-negative combinations of at most $d+1$ inequalities of \eqref{eq:system}, and combining them with the weights $\alpha_r$ gives $c^{\top}y\le\lambda^{\!\top}b\le\gamma$. This proves the first claim.

The finitely many inequalities obtained from the extreme rays are valid on $\pi_D(P)$ and imply every inequality valid on $\pi_D(P)$, in particular the inequalities of any description of $\pi_D(P)$, which exists by Corollary~\ref{cor:global}. They therefore describe $\pi_D(P)$. Two extreme rays cannot have the same support: if $\rho$ and $\rho'$ are non-proportional elements of $\Lambda_D$ with $\supp\rho=\supp\rho'$, then $\rho\pm\varepsilon(\rho'-\rho)\in\Lambda_D$ for all sufficiently small $\varepsilon>0$, so $\rho$ does not span an extreme ray. The number of extreme rays is therefore at most the number of non-empty subsets of $\{1,\dots,m\}$ with at most $d+1$ elements, which is $\sum_{k=1}^{d+1}\binom{m}{k}$. Finally, if $\pi_D(P)$ is full-dimensional then every description of it contains an inequality defining each of its facets, and a non-redundant description contains exactly one inequality per facet~\cite[Chapter~8]{schrijver1986}; so a non-redundant description has no more inequalities than the description by extreme rays.
\end{proof}

\end{document}